\documentclass[12pt]{article}
\usepackage{a4wide,amsfonts,amsmath,amssymb,amsthm,epsfig,enumerate}

\newcommand{\ecc}{{\mathrm{ecc}}}

\newcommand{\inst}[1]{$^{#1}$}

\newtheorem{thm}{Theorem}
\newtheorem{lem}{Lemma}
\newtheorem{prop}{Proposition}

\newtheorem{obs}{Observation}
\newtheorem{question}{Question}
\newtheorem{con}{Conjecture}
\newtheorem{problem}{Problem}
\theoremstyle{definition}

\newtheorem{rem}{Remark}

\begin{document}

\title{Modelling simultaneous broadcasting by level-disjoint partitions\footnote{This research was supported by the Czech Science Foundation grant GA14-10799S, ARRS Program P1-0383, and by ARTEMIS-JU project "333020 ACCUS".}}
\author{ \small Petr Gregor \inst{1}, \small Riste \v Skrekovski \inst{2,3}, and \small Vida Vuka{\v s}inovi{\'c}  \inst{4}}
\date{}
\maketitle
\begin{center}
{\footnotesize

\inst{1} Department of Theoretical Computer Science and Mathematical Logic, Charles University,\\
         Malostransk\'{e} n\'{a}m.~25, 11800 Prague, Czech Republic \\
        \texttt{gregor@ktiml.mff.cuni.cz}\\

\inst{2} Department of Mathematics, University of Ljubljana,\\
         Jadranska~19, 1000 Ljubljana, Slovenia\\

\inst{3} Faculty of Information Studies, \\
         Ljubljanska cesta~31A, 8000 Novo mesto, Slovenia\\
         \texttt{skrekovski@gmail.com}\\

\inst{4} Computer Systems Department, Jo{\v z}ef Stefan Institute, \\
         Jamova~39, 1000 Ljubljana, Slovenia\\
         \texttt{vida.vukasinovic@ijs.si}\\

}
\end{center}

\begin{abstract}
Simultaneous broadcasting of multiple messages from the same source vertex in synchronous networks is considered under restrictions that each vertex receives at most one message in a unit time step, every received message can be sent out only in the next time step, no message is sent to already informed vertex. The number of outgoing messages in unrestricted, messages have unit length, and we assume full-duplex mode.
In \cite{GSV} we developed a concept of level-disjoint partitions to study simultaneous broadcasting under this model. In this work we consider the optimal number of level-disjoint partitions. We also provide a necessary condition in terms of eccentricity and girth on existence of $k$ $v$-rooted level-disjoint partitions of optimal height. In particular, we provide a structural characterization of graphs admitting two level-disjoint partitions with the same root. 
\end{abstract}


%
%

\section{Introduction}

Broadcasting has been subject of extensive study and plays an important role in the design of communication protocols in various kinds of networks, see surveys~\cite{Grig,HHL,HKPRU}. The critical point for high performance computing is to assure efficient broadcasting, where data transmission is often found as the bottleneck. Here, the interconnection network is modeled as an undirected graph in which the vertices correspond to processors and the edges correspond to communication links between the processors. In this paper we restrict ourselves to synchronous networks, where at each time step messages can be sent from vertices to all their neighbors in unit time.

In \textit{broadcasting}, a single message located at one vertex has to be spread to all other vertices in the network. We study a more general variant, when several different messages are to be simultaneously transmitted from one source. This is motivated by a situation when large amount of data needs to be broadcasted in a network with bounded capacity of links. In this case data can be split into multiple messages and sent individually. Similar concepts for secure distribution of messages based on independent spanning trees were studied in~\cite{Chen,Cheng,Kim}.

The problem of same-source multiple broadcasting was considered previously under several different models~\cite{Noy,Haru}. 
In this paper we consider a communication model with restrictions that:

\begin{enumerate}[\hspace{0.5cm}(a)]
\item Different messages cannot arrive in the same vertex at the same time. Equivalently, each vertex can receive only one message in a unit time step (\emph{1-in port model}).
\item Every received message can be sent to neighbors only in the next step (\emph{no-buffer model}).
\item At every step, a message is sent only to vertices which have not received it yet (\emph{non-repeating model}).
\end{enumerate}

Note that although we restrict the number of incoming messages at each vertex to at most one, we do not limit the number of outgoing messages. That is, each vertex can send his message to all his neighbors at the same time (\emph{all-out port model}). We assume that messages have unit length; that is, each message (and only one message) can be sent through a link in a unit time step, and we further assume a so called \emph{full-duplex} mode, in which two adjacent vertices can exchange their messages through their link at the same time.

The condition (a) ensures that the load on vertices is minimized. Furthermore, it is useful for security reasons to avoid two messages meeting at a vertex at the same time, if untrustworthy vertex is able to decode secret information from having access to more messages at the same time. Similarly, an algorithm for secure message distribution with even more restricted condition when different messages can arrive in the same vertex only via distinct paths is proposed in~\cite{ChangYang}. The condition (b) can be useful for example in streaming context when every message needs to be sent immediately without delay. This is known as memoryless or queueless communication \cite{Jung}. The condition (c) ensures that (b) cannot be circumvented by sending message to a neighbor and then back. It also helps to prevent unnecessary network congestion and message delays.

Under the given model the following question arises: ``What is the minimal overall time needed for simultaneous broadcasting?''  Another questions are: ``What is the maximal number of messages that can be simultaneously broadcasted?'', and ``Is it always possible to broadcast them in the optimal time?'' Throughout the paper, \emph{optimal time} means meeting the lower bound given by eccentricity of the source vertex (see Observation~\ref{obs:c2}).

These questions were already studied. The time optimality for same-source simultaneous broadcasting were the subject of study in \cite{Noy,GSV,Haru}, while the maximum number of messages that can be simultaneously broadcasted via Hamiltonian cycles from a single vertex were studied in \cite{Sun, VGS}.

Let us remark that without restricting the number of messages that each vertex can receive at a time, a vertex may receive enormous amount of data altogether from its neighbors, which can cause a delay~\cite{ChangChen,Chang,FL}. The study of efficiency of the 1-in port models can be found in~\cite{Aver,Noy,Haru}.

Simultaneous broadcasting in the case when also the number of outgoing messages is limited to one (\emph{1-out port model}) has been considered previously in \cite{MI1,KW,KLH,Sun}. Observe that in the 1-out port non-repeating model each broadcasted message traverses a Hamiltonian path~\cite{Hsieh,LHHB}.


In the paper~\cite{GSV} we developed a concept of level-disjoint partitions to study how many messages and in what time can be simultaneously broadcasted under considered model from a given source vertex in a given graph. In this paper we show that the problem of simultaneous broadcasting in a graph $G$ can be solved locally on a suitable subgraph $H$ of $G$ and then extended to a solution for the whole graph $G$ (Lemma~\ref{lem:LDPdk}). 
We provide a structural characterization of graphs that admit simultaneous broadcasting of two messages from a given vertex (Theorem~\ref{thm:k=2}). In addition, we provide a necessary condition in terms of girth and eccentricity of the root for existence of $k$ same-rooted level-disjoint partitions of optimal height (Proposition~\ref{prop:cond1&2}).

\section{Level-disjoint partitions}

We use the following concept to capture the information dissemination in graphs according to our setting which was originally introduced in \cite{GSV}. A \emph{level partition} of a graph $G$ is a partition $\mathcal{S}=(S_0,\dots,S_h)$ of $V(G)$ into a tuple of sets, called \emph{levels}, such that
\begin{equation}\label{eq:ldp}
S_{i}\subseteq N(S_{i-1})
\end{equation}
for every $1\le i \le h$; that is, every vertex has a neighbor from previous level. The number $h=h(\mathcal{S})=|\mathcal{S}|-1$ is called the \emph{height} of $\mathcal{S}$, see an example in Figure~\ref{fig:ldps}(a).

\begin{figure}[h!]
\centerline{\includegraphics[scale=0.8]{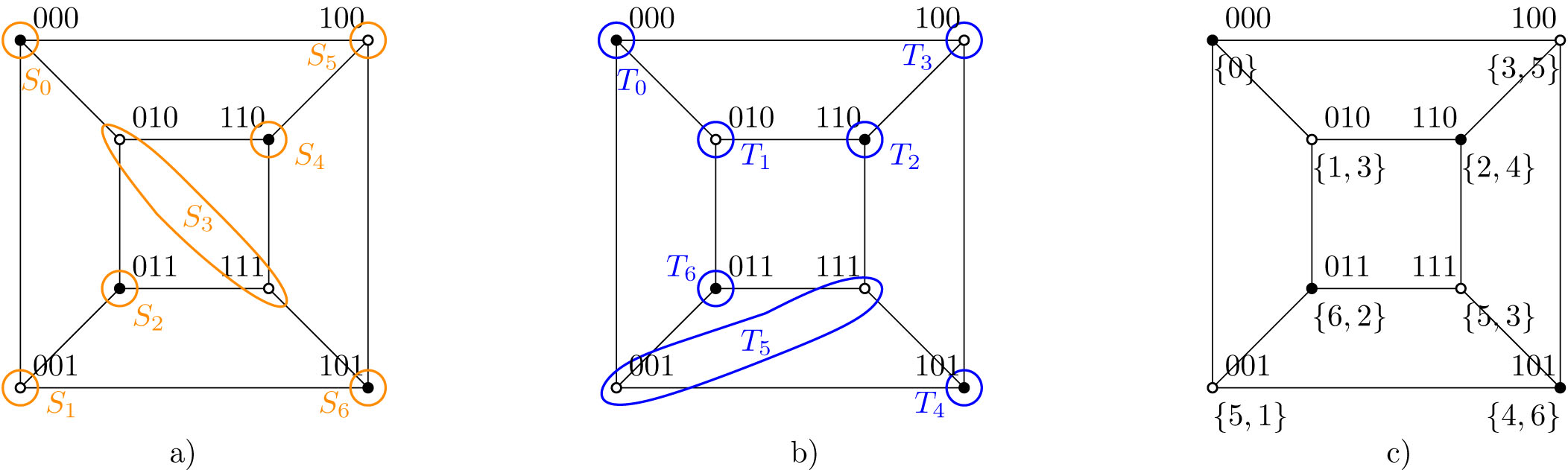}}
\caption{An example of two level-disjoint partitions of the hypercube $Q_3$ rooted in a vertex $v$. a) A level partition $\mathcal{S}=(\{000\},\{001\},\{011\},\{010,111\},\{110\},\{100\},\{101\})$. b) A level partition $\mathcal{T}=(\{000\},\{010\},\{110\},\{100\},\{101\},\{001,111\},\{011\})$. c) $R_{\mathcal{S}, \mathcal{T}}(u)$ of each vertex $u$.}\label{fig:ldps}
\end{figure}

This definition is motivated as follows. Starting at all vertices from the level $S_0$, at each step the same message is sent from all vertices of the current level to all vertices in the next level through edges of the graph. Thus $h$ is the number of steps needed to visit all vertices. Note that we abstract from which particular edges are used but the condition \eqref{eq:ldp} admits to choose edges so that each vertex (except the starting vertices) is informed by precisely one neighbor from a previous level (1-in port model). In this way, a message does not wait at vertices before it is sent further (no-buffer model), and since $\mathcal{S}$ is a partition of $V(G)$, no message returns to a previously visited vertex (non-repeating model).

Note that for the sake of generality we allow the same message to start at more than one vertex. In a particular case when the starting level $S_0$ is a singleton, say $S_0=\{v\}$, we say that the level partition is \emph{rooted} at $v$ (or $v$-\emph{rooted}) and the vertex $v$ is called the \emph{root} of $\mathcal{S}$. Rooted level partitions can be obtained, for example, from rooted spanning trees by taking the set of vertices at distance $i$ from the root $v$ in the spanning tree as the $i$-th level of the partition.

We are interested in a scenario when more messages should be broadcasted simultaneously with the requirement that they never arrive to the same vertex at the same time. This motivates the following definitions. Two level partitions $\mathcal{S}=(S_0,\dots,S_{h(\mathcal{S})})$ and $\mathcal{T}=(T_0,\dots,T_{h(\mathcal{T})})$ are said to be \emph{level-disjoint} if $S_i \cap T_i=\emptyset$ for every $1\le i \le \min(h(\mathcal{S}),h(\mathcal{T}))$. See an example on Figure~\ref{fig:ldps}(a) and (b). Note that we allow $S_0 \cap T_0 \neq \emptyset$ since we consider also the case when different messages have the same source (or overlapping sources). Level partitions $\mathcal{S}^1, \dots, \mathcal{S}^k$ are said to be (mutually) level-disjoint if each two partitions are level-disjoint. Then we say that $\mathcal{S}^1, \dots, \mathcal{S}^k$ are \emph{level-disjoint partitions}, shortly \emph{LDPs}. If every partition is rooted in the same vertex $v$ and they are level-disjoint (up to the starting level $\{v\}$), we say that $\mathcal{S}^1, \dots, \mathcal{S}^k$ are \emph{level-disjoint partitions with the same root $v$}, shortly \emph{$v$-rooted LDPs}.

Let $\mathcal{S}^1, \dots, \mathcal{S}^k$ be level partitions of $G$, not necessarily level-disjoint. The set of levels in which a given vertex $u$ occurs is denoted by $R_{\mathcal{S}^1, \dots, \mathcal{S}^k}(u)$; that is,
$$R_{\mathcal{S}^1, \dots, \mathcal{S}^k}(u)=\{l \mid u\in S_l^i\text{ for some }1\le i \le k\}.$$
We omit level partitions from the index of $R_{\mathcal{S}^1, \dots, \mathcal{S}^k}(u)$ if no ambiguity may arise. For an illustration see Figure~\ref{fig:ldps}(c).

Now, assume that $\mathcal{S}^1, \dots, \mathcal{S}^k$ have a common root $v$. Then clearly $R(v)=\{0\}$. Furthermore, they are level-disjoint if and only if $|R(u)|=k$ for every vertex $u$ except the root $v$. Finally, observe that if $G$ is bipartite, then for every vertex $u$ all elements (levels) in $R(u)$ have the same parity. More precisely, they are all even if $u$ and $v$ are from the same bipartite class; otherwise they are all odd.

The number of level-disjoint partitions determines how many messages can be broadcasted simultaneously while their maximal height determines the overall time of the broadcasting. Hence a general aim is to construct for a given graph
\begin{itemize}
    \item  as many as possible (mutually) level-disjoint partitions; and
    \item  with as small maximal height as possible.
\end{itemize}

Level-disjoint partitions of prescribed number are studied into the details in Section~\ref{sec:number}, while the optimality of their height in Section~\ref{sec:height}.

In this paper we consider only simple connected undirected graphs. Let $v$ be a vertex of a graph $G$. We denote by $G-v$ the graph obtained by removing $v$ and all incident edges from $G$. If a subgraph $H \subseteq G$ does not contain $v$, we denote by $H+v$ the subgraph of $G$ obtained by adding $v$ and all incident edges from $G$ to $H$. If $G-v$ is disconnected, the vertex $v$ is a \emph{cut-vertex}. A \emph{bridge} of $G$ is an edge whose removal disconnects $G$. A maximal subgraph without a cut-vertex is called a \emph{block}. Clearly, every block is $2$-connected, formed by a bridge, or an isolated vertex.

A \emph{path} is a (nonempty) sequence $P=(v_1,\dots,v_n)$ of distinct vertices with edges between consecutive vertices. The length of $P$ is $|P|=n-1$. We use notation $v_1Pv_n$ instead of $P$ if we need to point out the \emph{endvertices} $v_1,v_n$ of $P$. For vertex-disjoint (up to $v$) paths $uPv$ and $vP'w$ we denote by $(uPv, vP'w)$ their \emph{concatenation}. The \emph{reverse} of $P$ is the path $P^R=(v_n,\dots,v_1)$. Similarly, a \emph{cycle} is a sequence $C=(v_1,\dots,v_n)$ of distinct vertices up to $v_1=v_n$ with edges between consecutive vertices. The length of $C$ is $|C|=n-1$.
A \emph{girth} of $G$ is the length of a shortest cycle in $G$. The \emph{eccentricity} of $v$, denoted by $\ecc(v)$, is the maximal distance from $u$ to all vertices.

\section{Prescribed number of level-disjoint partitions}\label{sec:number}

In simultaneous broadcasting one often needs to send out at the same time as many messages as possible without limitation on overall time for a such task. In our communication model this scenario leads to the following question.
\begin{question}\label{q:numberOfLDPs}
Given a graph $G$, $v\in V(G)$, and $k\ge 1$, are there $k$ LDPs of $G$ rooted in $v$?
\end{question}

If we consider a problem of simultaneous broadcasting of $k$ messages from a single vertex $v$, in our setting, this corresponds to finding $k$ level-disjoint partitions of $G$ with the same root $v$. The obvious necessary condition on the number of $v$-rooted LDPs is as follows.

\begin{obs}\label{prop:c1c2} Let $\mathcal{S}^1,\dots,\mathcal{S}^{k}$ be level-disjoint partitions of a graph $G$ with the same root $v$. Then,
\begin{gather}
k \leq \deg(v).\label{eq:deg}
\end{gather}
\end{obs}
\begin{proof}
Note that the number of level-disjoint partitions $\mathcal{S}^1,\dots,\mathcal{S}^{k}$ starting in $v$ is at most $\deg(v)$ since $S_1^1, \ldots, S_1^k \subseteq N(v)$ and $S_1^i \cap S_1^j = \emptyset$ for every $i \neq j$. Hence \eqref{eq:deg} holds.
\end{proof}

We are interested in cases when the above bound is tight. If equality holds in \eqref{eq:deg} we say there is \emph{optimal number} of LDPs rooted in $v$.

If $v$ is a cut-vertex of $G$ and $G$ has $k$ LDPs rooted in $v$, then their restriction to each component $C_i$ of $G-v$ forms $k$ LDPs of $C_i+v$ rooted in $v$. Hence, $k \le \min_{i} \deg_{C_i+v}(v)$ where $\deg_G(v)=\sum_i \deg_{C_i+v}(v)$.
On the other hand, if there are $k$ $v$-rooted LDPs of $C_i+v$  for each component $C_i$ of $G-v$, they can be composed by taking unions of same levels into $k$ $v$-rooted LDPs of $G$. Hence the above question can be considered for each component of $G-v$ separately.

Furthermore, the following lemma shows that it suffices to find LDPs locally ``around'' the root $v$ on some suitable subgraph $H$ of $G$. Then they can be extended to LDPs with the same root to the whole graph $G$. Since $v$ could be a cut-vertex of $G$, we need that $H$ meets each component of $G-v$.

\begin{lem}\label{lem:LDPdk}
Let $v$ be a vertex of a graph $G$ and $H$ be a subgraph of $G$ containing $v$ and some vertex from each component of $G-v$. Then any $k$ $v$-rooted level-disjoint partitions of $H$ can be extended to $k$ $v$-rooted level-disjoint partitions of $G$.
\end{lem}
\begin{proof}
Let $\mathcal{S}^1,\ldots,\mathcal{S}^k$ be $v$-rooted level-disjoint partitions of $H$. If $V(H)=V(G)$ we are done, as they are $v$-rooted LDPs of $G$ as well. Now assume $V(H)\subsetneq V(G)$. It suffices to show that they can be extended to $v$-rooted level-disjoint partitions of $H'=H+u$
for some vertex $u$ of $G$ uncovered by $H$. Details are provided in the next paragraph. Then, by incremental extension until no uncovered vertex remains, we obtain $k$ $v$-rooted level-disjoint partitions of $G$.

Let $u$ be a vertex of $G$ that is not in $H$ but has a neighbor $w$ in $H$ distinct from $v$. Note that such vertex $u$ exists since $H$ contains some vertex from each component of $G-v$. Let us denote by $l_i$ the level of $w$ in $\mathcal{S}^i$; that is, $w\in S^i_{l_i}$ for every $1\le i \le k$. Then, we extend $\mathcal{S}^1,\ldots,\mathcal{S}^k$ to $H'$ by adding $u$ to the $(l_i+1)$-th level of $\mathcal{S}^i$ for every $1\leq i \leq k$.  Clearly, such extended partitions are level partitions of $H'$. Moreover, they are level-disjoint since $u$ was added into distinct levels of level-disjoint partitions $\mathcal{S}^1,\ldots,\mathcal{S}^k$.
\end{proof}

For $k=1$ the answer to the Question~\ref{q:numberOfLDPs} is trivial, since there is a level partition $(S_0,\dots,S_h)$ of $G$ starting in $v$ with $S_i=\{u\in V(G) \mid d_G(u,S_0)=i\}$ for every $0\le i \le h$, which is also of the optimal height. For $k=2$, it is easy to see that odd cycles have two LDPs with the same root whereas even cycles do not. For a full characterisation of graphs $G$ admitting two LDPs rooted in $v$, we need the following definitions.

A cycle containing a vertex $v$ is called a \emph{$v$-cycle}. Let us denote by $\overline{v}_C$ the opposite vertex to $v$ on an even cycle $C$. We say that a path $uPw$ is \textit{chordal} to a cycle $C$ if $V(P) \cap V(C) = \{ u, w \}$. Then we write $C=(vAu,uDw,wBv)$ to denote the subpaths $A,D,B$ of $C$ between respective vertices. We say that a chordal path $uPw$ to a cycle $C$ \textit{separates} $x,y \in V(C)$ if $x$ and $y$ belong to different subpaths of $C - \{u, w\}$. For an illustration see Figure~\ref{fig:chordal}(a).

\begin{figure}[h!]
\centerline{\includegraphics[scale=0.7]{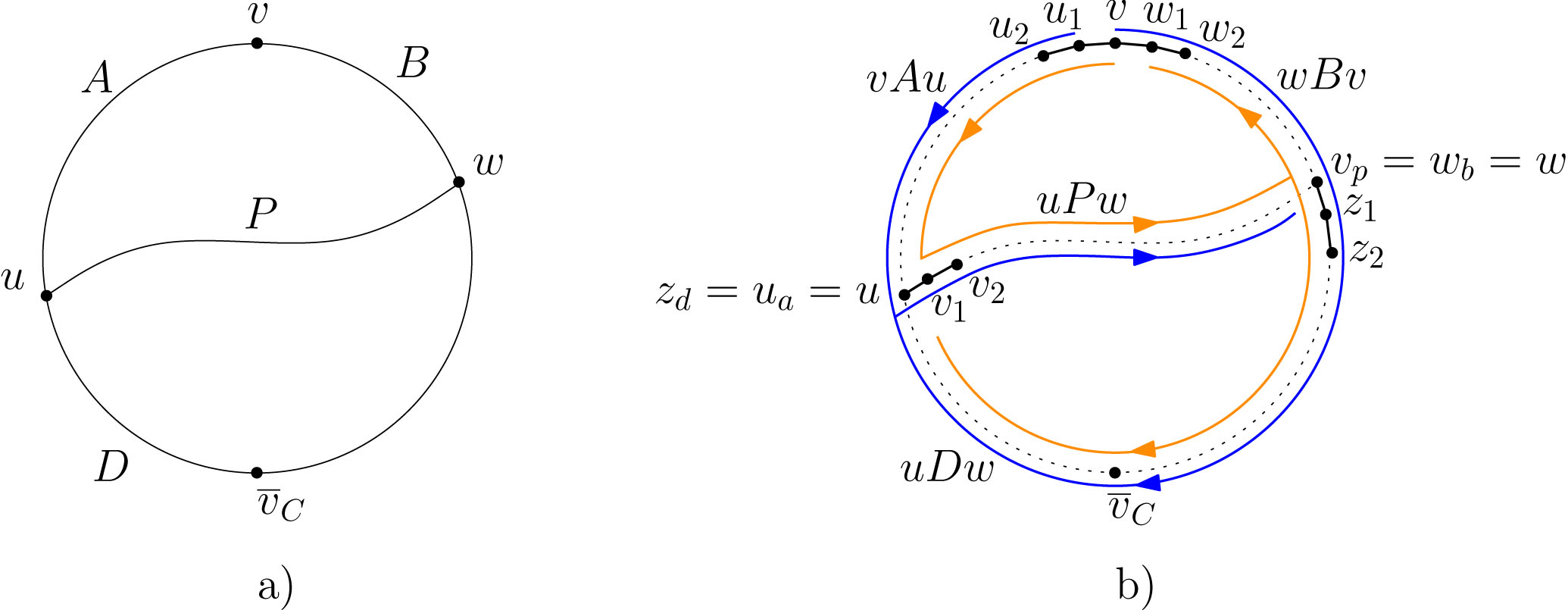}}
\caption{a) A $v$-cycle $C=(vAu,uDw,wBv)$ with a chordal path $uPw$ that separates $v$ and $\overline{v}_C$. b) A construction of two level-disjoint partitions $\mathcal{S}$, $\mathcal{T}$ of $C \cup P$ rooted at $v$.}\label{fig:chordal}
\end{figure}

\begin{obs}\label{obs:chordal}
Let $uPw$ be a chordal path to a $v$-cycle $C=(vAu,uDw,wBv)$, $a=|A|$, $b=|B|$, $d=|D|$ and assume that $a \geq b$. Then $P$ separates $v$, $\overline{v}_C$ if and only if $a< b+d$.
\end{obs}

\begin{proof}
Observe that $P$ separates $v$, $\overline{v}_C$ if and only if by moving from $v$ along the cycle $C$ we meet $\overline{v}_C$ after $u$ and before $w$. Since $\overline{v}_C$ is at distance $\frac{a+b+d}{2}$ from $v$ on $C$, this happens if and only if $a < \frac{a+b+d}{2} < a+d$. The second inequality holds by $d>0$ and the assumption $a\ge b$. Hence, $P$ separates $v$, $\overline{v}_C$ if and only if $a< b+d$.
\end{proof}

Let $P_1=(u_1,u_2,\dots,u_m)$, $P_2=(w_1,w_2,\dots,w_n)$ be paths with $u_1=w_n$ and $w_1=u_m$. Then $P_1$ and $P_2^R$ have a common prefix and a common suffix, we say that $P_1$, $P_2$ are \emph{merged}. Furthermore, if all common vertices of $P_1$, $P_2^R$ are in a common prefix or in a common suffix, we say that $P_1$, $P_2$ are \emph{fully-merged}; that is, they can be split into three subpaths $P_1=(A,P,B)$, $P_2=(B^R,D^R,A^R)$ with $P,D$ possibly being both empty such that $A$ is the longest common prefix of $P_1,P_2^R$, $B$ is the longest common suffix of $P_1,P_2^R$, and $P,D$ are inner vertex-disjoint.

\begin{lem}\label{lem:fully}
Let $P_1=(A,P,B)$, $P_2=(B^R,D^R,A^R)$, where $|P|\le|D|$, be fully-merged level-disjoint paths in a bipartite graph $G$ such that their endvertices have a common neighbor $v$ not belonging to $P_1 \cup P_2$. Then $P$ is a chordal path on the cycle $C=(v,A,D,B,v)$ separating $v, \overline{v}_C$.
\end{lem}
\begin{proof}Since $P_1$, $P_2$ are level-disjoint, also the paths $P'_1=(v,P_1)$, $P'_2=(v,P_2)$ are level-disjoint. We claim that both $P$ and $D$ are nonempty. Clearly, if one of $P$, $D$ is empty, then the other is empty as well, for otherwise $P_1$ or $P_2$ is not a path. If both $P$ and $D$ are empty, then $P'_1$, $P'_2$ form two $v$-rooted LDPs of a cycle $C=(v,A,B,v)$. But this is impossible, since $C$ is even by the assumption that $G$ is bipartite, and even cycles do not have two same-rooted LDPs as we observed above. Hence, the claim holds, so $P$ is a chordal path on a cycle $C=(v,A,D,B,v)$.

What remains is to show that $P$ separates $v$, $\overline{v}_C$ on $C$. Let $a=|A|+1$, $b=|B|+1$, $d=|D|$. Note that we assume above that $d\ge |P|$ and we may also assume that $a\ge b$ for otherwise we can reverse paths $P_1$, $P_2$. Since $P'_1=(v,A,P,B)$, $P'_2=(v,B^R,D^R,A^R)$ are level-disjoint, every vertex of $A=(u_1,\dots,u_a)$ has distinct levels in $P'_1$, $P'_2$. More precisely, the vertex $u_i$ in the $i$-th level on $P'_1$ and in the $(a+b+d-i)$-th level on $P'_2$. Note that $a+b+d=|C|$ is even since $G$ is bipartite. Suppose for a contradiction that $a\ge b+d$. Then $i=\frac{a+b+d}{2}\le a$, so $u_i$ belongs to $A$. But this implies that $u_i$ is in the $i$-th level of both $P'_1$ and $P'_2$ contrary to their level-disjointness. Therefore $a<b+d$, and by Observation~\ref{obs:chordal} it follows that $P$ separates $v, \overline{v}_C$ on $C$.
\end{proof}

Now we are ready to characterize graphs admitting two level-disjoint partitions rooted in a given vertex.

\begin{thm}\label{thm:k=2}
Let $v$ be a vertex in a graph $G$. Then $G$ has two level-disjoint partitions rooted in $v$ if and only if for every block $B$ of $G$ containing $v$ it holds that
\begin{enumerate}[\hspace{0.5cm}$(a)$]
  \item $B$ is $2$-connected, and
  \item $B$ is non-bipartite or $B$ has a $v$-cycle $C$ with a chordal path that separates $v$ and $\overline{v}_C$.
\end{enumerate}
\end{thm}

\begin{proof}
Let $C_1, \ldots, C_l$ be components of $G - v$ and let $B_1, \ldots, B_l$ be blocks of $G$ containing $v$ such that $B_i \subseteq C_i +v$ for every $1 \leq i \leq l$. First, to prove the implication from right to left (sufficiency), assume that every block $B_i$ is $2$-connected, and $B_i$ is non-bipartite or $B_i$ has a $v$-cycle $C$ with a chordal path that separates $v$ and $\overline{v}_C$.

Clearly, if every $B_i$ has two level-disjoint partitions rooted in $v$, they can be extended to two level-disjoint partitions of $C_i+v$ by Lemma~\ref{lem:LDPdk}, and then composed to two level-disjoint partitions of $G$. Hence, it is sufficient to find two $v$-rooted level-disjoint partitions in every $B_i$.

If $B_i$ is non-bipartite, we proceed as follows.
The vertex $v$ belongs to some odd cycle $C= (v=u_0, \ldots, u_{2k},v)$ in $B_i$, since in any $2$-connected non-bipartite graph, every vertex belongs to an odd cycle.
Then, it is easy to see that
\begin{align*}&(v=u_0, u_1, u_2, \ldots, u_k, u_{k+1}, \ldots, u_{2k-1},u_{2k}),\\
&(v=u_0, u_{2k}, u_{2k-1}, \ldots, u_{k+1}, u_{k}, \ldots, u_2, u_{1})\end{align*}
are level-disjoint paths and thus they form two $v$-rooted LDPs of $C$. By applying Lemma~\ref{lem:LDPdk} for a subgraph $C$ of $B_i$ we obtain two level disjoint partitions of $B_i$.

Assume now that the block $B_i$ is bipartite.
Let $C=(vAu,uDw,wBv)$ be a $v$-cycle in $B_i$ with a chordal path $uPw$ that separates $v$ and $\overline{v}_C$. For $a=|A|$, $b=|B|$, $d=|D|$, $p=|P|$ we may assume that $a \geq b$ and $d\geq p$, see Figure~\ref{fig:chordal}b), and let us denote the vertices in each subpath by
\begin{align*}A&=(v,u_1,u_2,\dots,u_a=u),&  B&=(w=w_b,w_{b-1},\dots, w_1,v),\\
D&=(u=z_d,z_{d-1},\dots,z_1,w),&  P&=(u,v_1,v_2,\dots, v_p=w).
\end{align*}
By Observation~\ref{obs:chordal}, it holds that $a< b+d$. We define level partitions $\mathcal{S}=(S_0,\dots,S_{h_S})$, $\mathcal{T}=(T_0,\dots,T_{h_T})$, where $h_S=\max(a+d+p-1, a+b+p-1)$ and $h_T=\max(a+b+d-1, b+d+p-1)$, of $C \cup P$ rooted at $v$ as follows:
\begin{displaymath}
S_i = \left\{ \begin{array}{ll}
\{v\} & \hbox{for }i=0,\\
\{u_i\} & \hbox{for }i=1, \ldots, a, \\
\{v_{i-a}\} & \hbox{for }i=a+1, \ldots, a+p, \\
\{z_{i-a-p}, w_{a+b+p-i}\} & \hbox{for }i=a+p+1, \ldots, h_S,
\end{array} \right.
\end{displaymath}

\begin{displaymath}
T_i = \left\{ \begin{array}{ll}
\{v\} & \hbox{for }i=0,\\
\{w_i\} & \hbox{for }i=1, \ldots, b, \\
\{z_{i-b}\} & \hbox{for }i=b+1, \ldots, b+d, \\
\{u_{a+b+d-i}, v_{i-b-d}\} & \hbox{for }i=b+d+1, \ldots, h_T.
\end{array} \right.
\end{displaymath}
In the above definitions if $u_j$, $v_j$, $w_j$, or $z_j$ is not defined, it is not included in $S_i$, $T_i$. (For example, if $b\le d$ then $h_S=a+d+p-1$ and $w_0$ is not included in $S_{a+b+p}$.)

The vertex $u_i$ is in different levels of $\mathcal{S}$, $\mathcal{T}$ since $a<b+d+1$. Similarly, $w_i$ is in different levels of $\mathcal{S}$, $\mathcal{T}$ since $b<a+p+1$. Furthermore, $z_i$ is in different levels of $\mathcal{S}$, $\mathcal{T}$ since $b<a+p$. Similarly, $v_i$ is in different levels of $\mathcal{S}$, $\mathcal{T}$ since $a < b+d$. Hence, $\mathcal{S}$ and $\mathcal{T}$ are level-disjoint. By Lemma~\ref{lem:LDPdk}, $\mathcal{S}$ and $\mathcal{T}$ can be extended to level-disjoint partitions of $B_i$.

Second, for the other implication (necessity), let us assume that $G$ has two level-disjoint partitions $\mathcal{S}$ and $\mathcal{T}$ and let $B$ be a block of $G$ containing the vertex $v$. Note that $B$ has two level-disjoint partitions $\mathcal{S}'$ and $\mathcal{T}'$ induced by $\mathcal{S}$ and $\mathcal{T}$, since every vertex $u \in V(B)$ at level $S_i$ ($T_i$) has some neighbor in $B$ at level $S_{i-1}$ (respectively $T_{i-1}$). Since $B$ has two level disjoint partitions with the starting level $\{v\}$, $B$ is $2$-connected.

We need to show that if $B$ is bipartite, then $B$ has a $v$-cycle $C$ with a chordal path that separates $v$ and $\overline{v}_C$. So assume that $B$ is bipartite. Since $B$ has level-disjoint partitions $\mathcal{S}'$ and $\mathcal{T}'$, there exist two level-disjoint paths $(v,w_n=u_1, u_2, \ldots, u_m=w_1)$, $(v,u_m=w_1, w_2, \ldots, w_n=u_1)$ from $v$ to its neighbors $w_1$, $u_1$ in $B$, respectively. Hence $P_1=(u_1,\dots,u_m)$, $P_2=(w_1,\dots, w_n)$ are level-disjoint as well and moreover merged. We claim that $P_1$, $P_2$ can be adjusted into fully-merged level-disjoint paths between the same vertices. We proceed by replacing a suffix of one path by a prefix of the other path as described by the following steps:

\begin{itemize}
  \item Assume that $P_1$, $P_2$ are not fully-merged yet. Then there is a vertex $u$ in $V(P_1) \cap V(P_2)$ that is neither in a common prefix nor in a common suffix of $P_1$, $P_2^R$. Let $l_1$ be the smallest index of such vertex $u$ in $P_1$ and let $l_2$ be the index of $u$ in $P_2$; that is, $u=u_{l_1}=w_{l_2}$ and $u_{l_1-1}\notin V(P_2)$. Clearly, $1<l_1<m$ and $1<l_2<n$ as $P_1$, $P_2$ are merged. Moreover, $l_1\ne l_2$ as $P_1$, $P_2$ are level-disjoint.
  \item If $l_1 < l_2$ then put $P_1' := P_1$ and $P_2' := (w_1, w_2, \ldots, w_{l_2}=u_{l_1}, u_{l_1-1}, \ldots, u_1=w_n)$.
  \item If $l_1 > l_2$ then put $P_1' := (u_1, u_2, \ldots, u_{l_1}=w_{l_2}, w_{l_2-1}, \ldots, w_1=u_n)$ and $P_2' := P_2$.
  \item Repeat the above steps for paths $P_1', P_2'$ instead of $P_1, P_2$ until they are fully merged.
\end{itemize}

We claim that $P'_1$, $P'_2$ are level-disjoint. If $l_1<l_2$ then the prefix $(w_1,w_2,\dots,w_{l_2})$ of $P'_2$ is level-disjoint with $P'_1=P_1$ since $P_1$, $P_2$ are level-disjoint. Moreover, the remaining vertices of $P'_2$; that is, $u_1,\dots,u_{l_1}$ are in levels up to $l_1$ in $P'_1$ and in levels from $l_2$ in $P'_2$. Hence, each vertex is in different levels on $P'_1$ and $P'_2$. If $l_1>l_2$ then the prefix $(u_1,u_2,\dots,u_{l_1})$ of $P'_1$ is level-disjoint with $P'_2=P_2$ since $P_1$, $P_2$ are level-disjoint. Moreover, the remaining vertices of $P'_1$; that is, $w_1,\dots,w_{l_2}$ are in levels up to $l_2$ in $P'_2$ and in levels from $l_1$ in $P'_2$. Hence, each vertex is in different levels on $P'_1$ and $P'_2$, so the claim holds.

Clearly, $P'_1$, $P'_2$ have the same endvertices as $P_1$, $P_2$, respectively. Furthermore, $P'_1$, ${P'}^R_2$ have a longer common prefix or a longer common suffix than $P_1$, $P^R_2$. Hence the above process ends after finitely many steps and produces fully-merged level-disjoint paths between the same vertices as desired. Therefore, by Lemma~\ref{lem:fully}, $B$ has a $v$-cycle $C$ with a chordal path that separates $v$, $\overline{v}_C$.
\end{proof}

If a large amount of data needs to be transmitted from a given vertex to all other vertices in the network with bounded capacity of links, the transmission could be simultaneously done by following as much as possible level-disjoint partitions from the first level to the last level. Therefore, it would be interesting to find a similar structural characterization for existence of $k$ LDPs with the same root as in Theorem~\ref{thm:k=2} also for $k\ge 3$. More formally, it can be stated as follows.

\begin{problem}
Find a structural characterization of graphs admitting $k$ LDPs rooted in a given vertex $v$.
\end{problem}
\noindent Moreover, we propose the following conjecture.

\begin{con}
In a $k$-connected graph ($k \geq 3$) with at least $2k+1$ vertices, there exist $k$ level-disjoint partitions rooted at any given vertex $v$.
\end{con}

Observe that the conjecture does not hold for the graphs of smaller order, since for any given vertex $v$ in $K_{k,k}$ there are only $k-1$ vertices at the distance $2$ from vertex $v$. We are not aware of any other such graph.

For an example of three level-disjoint partitions of a $3$-connected graph and four level-disjoint partitions of a $4$-connected graph see Figure~\ref{fig:kconnected}(a) and (b), respectively.

\section{Level-disjoint partitions of optimal height}\label{sec:height}

In other scenario of simultaneous broadcasting one needs to send out multiple messages with a guarantee that it will take a shortest time possible. This leads to the following question.

\begin{question}
Given a graph $G$, $v\in V(G)$, and $k\ge 2$, what is the smallest maximal height as possible of $k$ LDPs of $G$ rooted in $v$?
\end{question}

First we consider a problem of simultaneous broadcasting of $k$ messages from a single vertex $v$. In our setting, this corresponds to finding $k$ level-disjoint partitions of $G$ with the same root $v$. We start with some obvious necessary conditions on the maximal height of $v$-rooted LDPs.

Assume that $\mathcal{S}^1, \dots, \mathcal{S}^k$ are $v$-rooted LDPs of $G$. Clearly, for every vertex $u$ except $v$,
\begin{align}
\min(R(u))&\ge d(u,v)\quad{\hbox{and}}\label{eq:min}\\
\max(R(u))&\ge d(u,v)+k-1\label{eq:max}
\end{align}
since $u$ cannot appear in a level smaller than the distance to the root $v$ and $|R(u)|=k$. If equality holds in $\eqref{eq:max}$ (and thus also in \eqref{eq:min}), then,
$$R(u)=\{d(u,v),d(u,v)+1,\dots,d(u,v)+k-1\}.$$
This means that all $k$ messages will be delivered to the vertex $u$ in the best time possible for this vertex.

If $G$ is bipartite, then for any same-rooted LDPs of $G$, the $R(u)$ of each vertex $u$ contains elements of the same parity. It follows that for no vertex equality in~\eqref{eq:max} holds (except in the trivial case of a single partition).
If $G$ is bipartite, we may strengthen \eqref{eq:max} by
\begin{equation}
\max (R(u))\ge d(u,v)+2k-2.\label{eq:maxb}
\end{equation}
For bipartite $G$, if equality holds in $\eqref{eq:maxb}$, then,
$$R(u)=\{d(u,v),d(u,v)+2,\dots,d(u,v)+2k-2\}.$$
Necessary conditions on same-rooted LDPs are as follows.

\begin{obs}\label{obs:c2} Let $\mathcal{S}^1,\dots,\mathcal{S}^{k}$ be level-disjoint partitions of a graph $G$ with the same root $v$. Then,
\begin{gather}
\max_{1 \leq i \leq k}h(\mathcal{S}^i) \geq \begin{cases}
\ \ecc(v)+k-1 & \textrm{if $G$ is not bipartite},\\
\ \ecc(v)+2k-2 & \textrm{if $G$ is bipartite}.\end{cases}\label{eq:hei}
\end{gather}
\end{obs}
\begin{proof}
Let $u$ be an eccentric vertex to $v$; that is, $d(u,v) = \ecc(v)$. By \eqref{eq:max} and \eqref{eq:maxb},
$$\max(R(u))\ge \left\{ \begin{array}{ll}
\ecc(v)+k-1 & \textrm{if $G$ is not bipartite},\\
\ecc(v)+2k-2 & \textrm{if $G$ is bipartite}.
\end{array} \right.$$
Since by definition,
$\max_{1 \leq i \leq k}h(\mathcal{S}^i) = \max_{u \in V(G)} \max(R(u))$,
it follows that also \eqref{eq:hei} holds.
\end{proof}

We are interested in cases when the above bound is tight. If equality holds in \eqref{eq:hei} we say that $\mathcal{S}^1, \dots, \mathcal{S}^k$ have \emph{optimal height}. Observe that if equality in~\eqref{eq:max} or equality in~\eqref{eq:maxb} for bipartite graphs holds for all vertices (up to the root vertex) in LDPs, then LDPs have optimal height. See examples on Figure~\ref{fig:kconnected}(b) and on Figure~\ref{fig:bldps}.

\begin{figure}[h!]
\centerline{\includegraphics[scale=0.8]{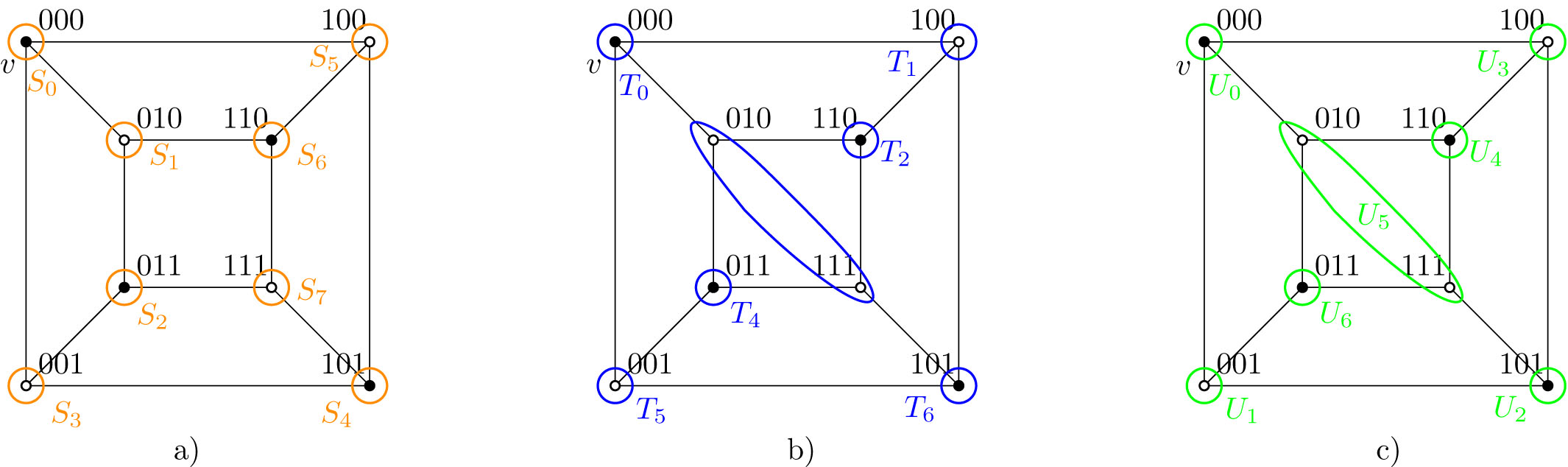}}
\caption{An example of three level-disjoint partitions of the hypercube $Q_3$ rooted in a vertex $v$ in which the inequality~\eqref{eq:maxb} holds for all vertices. a) A level partition $\mathcal{S}=(\{000\},\{010\},$ $\{011\},\{001\},\{101\},\{100\},\{110\},\{111\})$. b) A level partition $\mathcal{T}=(\{000\},$ $\{100\},\{110\},\{010,111\},\{011\},\{001\},\{101\})$. c) A level partition $\mathcal{U}=(\{000\},$ $\{001\},\{101\},\{100\},\{110\},\{010,111\},\{011\})$.}\label{fig:bldps}
\end{figure}

In addition, we provide a necessary condition in terms of girth.
\begin{prop}\label{prop:cond1&2}
Let $G$  be a graph of girth $s$ having $k\ge 2$ $v$-rooted level-disjoint partitions of optimal height. Then,
\begin{equation}\label{eq:ecc}
\ecc(v) \geq \begin{cases}\ s-2 &\text{if $G$ is non-bipartite,}\\
\ s-3 &\text{if $G$ is bipartite.}\end{cases}
\end{equation}
\end{prop}

\begin{proof}
Let $\mathcal{S}^1,\dots,\mathcal{S}^{k}$ be level-disjoint partitions of $G$ rooted in $v$ and let $u$ be a vertex from the first level of some partition. Note that $u$ is a neighbor of $v$. Since $\mathcal{S}^1,\dots,\mathcal{S}^{k}$ are level-disjoint, $u$ belongs to distinct levels, and denote these levels by $l_1=1,l_2,\ldots,l_k$ increasingly ordered. Since $G$ has girth $s$, we obtain that $l_2\ge s-1$.

Assume first that $G$ is non-bipartite. Since $s-1 \le l_2<l_3<\cdots <l_k$, we conclude that $l_k\ge s+k-3$. From the other side, the assumption on the optimal height gives us $l_k\le \ecc(v) + k -1$. Hence,
$s+k-3 \leq \ecc(v) + k -1$,
which implies \eqref{eq:ecc}.

If $G$ is bipartite, from $s-1\le l_2$ and $l_i + 2 \le l_{i+1}$ for each $i\ge 1$, we conclude that $l_k\ge s + 2 k -5$. The assumption on the optimal height gives us $l_k\le \ecc(v) + 2k -2$. Hence,
$s + 2 k -5\le \ecc(v) + 2k -2$,
which again implies \eqref{eq:ecc}.
\end{proof}
\begin{rem}Note that Proposition~\ref{prop:cond1&2} can be strengthened for local girth of the root instead of (global) girth of $G$. A \emph{local girth} of a vertex $v$ is the minimal length of a $v$-cycle.
\end{rem}

For example, consider Petersen graph, which is a 3-regular vertex-transitive non-bipartite graph of girth $5$ and eccentricity $2$ at every vertex. Note that condition~\eqref{eq:ecc} does not hold for any vertex which implies that Petersen graph has neither $2$ nor $3$ level-disjoint partitions of optimal height.

\begin{figure}[h!]
\centerline{\includegraphics[scale=0.8]{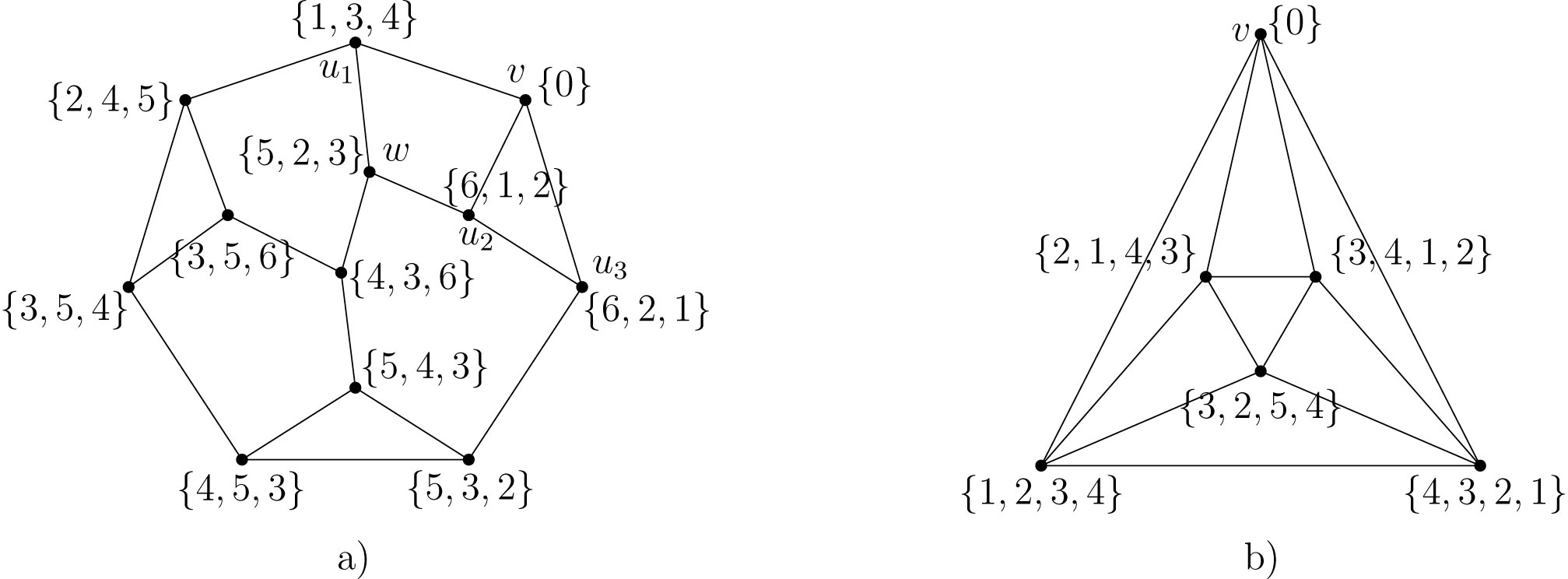}}
\caption{Examples of optimal number of LDPs of $3$-connected and $4$-connected graphs. a) Three $v$-rooted LDPs of a $3$-connected graph. b) Four $v$-rooted LDPs of optimal height of a $4$-connected graph.}\label{fig:kconnected}
\end{figure}
On the other hand, the condition from Proposition~\ref{prop:cond1&2} is necessary but it is not sufficient, what is evident from an example in Figure~\ref{fig:kconnected}(a). Let $G$ be a graph from Figure~\ref{fig:kconnected}(a). Observe that $G$ is non-bipartite, has girth $3$, and $\ecc(v) = 2$. Let $\mathcal{S},\mathcal{T},\mathcal{U}$ be level-disjoint partitions rooted in a vertex $v$ and let $u_1,u_2,u_3$ be vertices of $S_1, T_1, U_1$, respectively. Note, that $w$ cannot belong to both $S_2, T_2$. It is easy to see that, if $w \in T_2$, then $R_{\mathcal{S}}(u_3) \geq 6$. Similarly, if $w \in S_2$, then $R_{\mathcal{T}}(u_1) \geq 7$. Hence, $G$ has no three LDPs of optimal height rooted in a vertex $v$.

\section{Conclusion and further work}

Simultaneous broadcasting of multiple messages from the same source in the considered communication model is appropriately captured by the concept of same-rooted level-disjoint partitions of graphs. It was originally introduced in \cite{GSV} and further developed here.

In the context of broadcasting, the same concept can be employed to describe simultaneous broadcasting of multiple messages from distinct sources, possibly with each message having multiple originators. This can be subject of further research.

In Theorem~\ref{thm:k=2} we characterized graphs admitting two level-disjoint partitions rooted in a given vertex. It would be interesting to find a similar characterization also for a larger number of level-disjoint partitions.

As a further work in~\cite{GSVb} we study weather a local solution on a suitable subgraph can be without loss of optimality extended to the whole graph and we specify subgraphs which lead to simultaneous broadcasting in optimal time. In this context we study subgraphs of Cartesian products, in particular, bipartite tori, meshes, and hypercubes. 

Finally, one can further study the same problem under additional assumptions on the communication model. In this paper we assumed the all-out port model and the full-duplex mode. It would be interesting to obtain similar results also when the number of outgoing messages is restricted and/or for the half-duplex mode.

\end{document}